\documentclass[journal,draftclsnofoot,onecolumn,12pt,twoside]{IEEEtran}
\usepackage{amsmath,amsfonts}
\usepackage{algorithmic}
\usepackage{algorithm}
\usepackage{amssymb}

\usepackage{color}

\usepackage{float} 

\usepackage{array}
\usepackage[caption=false,font=normalsize,labelfont=sf,textfont=sf]{subfig}
\usepackage{textcomp}
\usepackage{stfloats}
\usepackage{url}
\usepackage{bm}
\usepackage{setspace}
\usepackage{verbatim}
\usepackage{multirow} 
\usepackage{ntheorem}
\usepackage{graphicx}
\usepackage{cite}
\usepackage{amsmath}

\allowdisplaybreaks[4]

\theorembodyfont{\upshape}
\theoremheaderfont{\it}
\qedsymbol{\square}
\theoremseparator{:}
\newtheorem{theorem}{\indent Theorem}
\newtheorem{lemma}{\indent Lemma}
\newtheorem*{proof}{\indent Proof}
\newtheorem{remark}{\indent Remark}

\makeatletter

\newcommand{\Rmnum}[1]{\expandafter\@slowromancap\romannumeral #1@}
\makeatother
\begin{document}

\title{Imperfect CSI: A Key Factor of Uncertainty to Over-the-Air Federated Learning}

\author{Jiacheng~Yao,~\IEEEmembership{Graduate~Student~Member,~IEEE},
		Zhaohui~Yang,~\IEEEmembership{Member,~IEEE},
		Wei~Xu,~\IEEEmembership{Senior~Member,~IEEE},
		Dusit~Niyato,~\IEEEmembership{Fellow,~IEEE},
		and Xiaohu~You,~\IEEEmembership{Fellow,~IEEE}
		
\thanks{J. Yao, W. Xu, and X. You are with the National Mobile Communications Research Laboratory (NCRL), Southeast University, Nanjing 210096, China (\{jcyao, wxu, xhyu\}@seu.edu.cn).}
\thanks{Zhaohui Yang is with the Zhejiang Lab, Hangzhou 311121, China, and also with the College of Information Science and Electronic Engineering, Zhejiang University, Hangzhou, Zhejiang 310027, China (yang\_zhaohui@zju.edu.cn).}
\thanks{Dusit Niyato is with the School of Computer Science and Engineering, Nanyang Technological University, Singapore 308232 (dniyato@ntu.edu.sg).}
}

%

\maketitle

\begin{abstract}
Over-the-air computation (AirComp) has recently been identified as a prominent technique to enhance communication efficiency of wireless federated learning (FL). This letter investigates the impact of channel state information (CSI) uncertainty at the transmitter on an AirComp enabled FL (AirFL) system with the truncated channel inversion strategy. To characterize the performance of the AirFL system, the weight divergence with respect to the ideal aggregation is analytically derived to evaluate learning performance loss. We explicitly reveal that the weight divergence deteriorates as $\mathcal{O}(1/\rho^2)$ as the level of channel estimation accuracy $\rho$ vanishes, and also has a decay rate of $\mathcal{O}(1/K^2)$ with the increasing number of participating devices, $K$. Building upon our analytical results, we formulate the channel truncation threshold optimization problem to adapt to different $\rho$, which can be solved optimally. Numerical results verify the analytical results and show that a lower truncation threshold is preferred with more accurate CSI.

\end{abstract}


\section{Introduction}
\IEEEPARstart{F}{ederated} learning (FL), a distributed machine learning paradigm, has been regarded as a promising technique to support ubiquitous intelligence in the beyond fifth-generation (B5G) wireless networks \cite{xu, push}. In a wireless FL system, the distributed devices, orchestrated by a parameter server (PS), iteratively train a shared learning model through the exchange of model parameters rather than the raw data, thereby protecting data privacy \cite{energy, ajoint}. However, due to the frequent uplink transmissions of model parameters from a large number of devices, the communication overhead and latency of FL become excessively high, which hinders its deployment in resource-constrained wireless networks.

To facilitate communication-efficient FL design, over-the-air computation (AirComp) has been greatly adopted for effective uplink model transmission \cite{zhu,zhu2,guo}. By exploiting the waveform superposition nature of multiple access (MAC) channels, simultaneous model transmission and over-the-air model aggregation can be achieved, which can reduce the communication latency and save the uplink communication bandwidth substantially. In \cite{zhu}, a truncated channel inversion scheme was proposed to combat deep fadings in an AirComp-aided FL (AirFL), and the fundamental trade-offs between communication and learning was discussed. Then in \cite{zhu2}, the power control strategy was further optimized to alleviate the impacts brought by AirComp errors. Considering the constraint of limited wireless communication resources, device selection and power control were jointly optimized to minimize the accuracy loss for AirFL in \cite{guo}.

However, most of the existing AirFL scheme design and resource allocation optimization optimistically assumed the availability of perfect channel state information (CSI) at the transmitter, which is hardly to acquire in practice especially in a wireless network. 
More importantly, unlike traditional communication systems, the CSI imperfection in the AirFL system brings a severe impact. To be concrete, considering transmit power constraints, the users in deep fading should be truncated and therefore not participate in AirComp. Moreover, to achieve the uniform model aggregation, channel inversion should be performed at the transmitter.
Considering the CSI uncertainties, the model aggregation is perturbed due to the imperfect  truncation decision and  channel inversion, resulting in the deterioration in learning performance. In \cite{onebit}, the authors considered a bounded CSI error and analyzed the impact of imperfect CSI on the convergence rate of FL. However, few effort has been endeavored to explicitly analyze in theory the aggregation distortion and accuracy loss brought by CSI uncertainty. To the best of our knowledge, there is no theoretical guidance on channel truncation strategy of imperfect CSI.

Against the above backgrounds, we focus on an AirFL system adopting the truncated channel inversion scheme, where only partial CSI is available at the PS. We theoretically characterize the aggregation distortion due to imperfect CSI and the corresponding channel truncation, and derive an upper bound of the weight divergence of the aggregated gradient to evaluate the degradation of learning performance. Our results unrevil that as the level of channel estimation accuracy $\rho$ decreases, the weight divergence enlarges at the order of $1/\rho^2$. The analytical result further suggests that increasing the number of participating devices, $K$, help decrease the weight divergence as $\mathcal{O}(1/K^2)$ and can completely eliminate the impact of CSI imperfection.
Moreover, based on the derived analytical results, we derive the optimal truncation threshold as a function of channel estimation uncertainty and system SNR. Numerical tests are conducted to verify the effectiveness of performance analysis and truncation threshold optimization.

\section{System Model of AirFL}

\subsection{Federated Learning Model}
We consider a typical FL algorithm, where a shared machine learning model is trained via the collaboration between a central PS and $K$ distributed devices. Let $\mathcal{D}_k$ denote the local dataset owned by the $k$th device.  The local loss function of model parameters, $\bm{w}$, at the device $k$ is defined as
\begin{align}
F_k(\bm{w},\mathcal{D}_k)=\frac{1}{\left \vert \mathcal{D}_k\right \vert}\sum_{\bm{u}\in \mathcal{D}_k} \mathcal{L}(\bm{w},\bm{u}),
\end{align}
where $\bm{u}$ is a data sample and $\mathcal{L}(\bm{w},\bm{u})$ represents the sample-wise loss function. Without loss of generality, we assume that the size of all local datasets is the same, i.e., $\left \vert \mathcal{D}_k\right \vert=D$, $\forall k$. Then, the global loss function over all the datasets is given by
\begin{align}\label{e1}
F(\bm{w})=\frac{1}{K}\sum_{k=1}^K  F_k(\bm{w},\mathcal{D}_k).
\end{align}
The goal of the FL is to find the optimal model parameters, denoted by $\bm{w}^*$, to minimize the global loss function in (\ref{e1}). 

To effectively handle this problem, we apply the widely used FL algorithm in \cite{onebit}.  Specifically, in the $m$th round of the FL algorithm, the PS firstly broadcasts the up-to-date global parameter $\bm{w}_{m}$ to all devices. With the received global model $\bm{w}_{m}$ and their local datasets, each device runs a stochastic gradient descent (SGD) algorithm on a local mini-batch to compute the local gradient, which follows
\begin{align}\label{e3}
\bm{g}_{m}^k\triangleq  \nabla F_k\left(\bm{w}_m,\mathcal{D}_{k,m}\right)=\frac{1}{\left\vert\mathcal{D}_{k,m} \right \vert}\sum_{\bm{u}\in \mathcal{D}_{k,m}} \mathcal{L}(\bm{w}_m,\bm{u}),
\end{align}
where $\mathcal{D}_{k,m}$ is the mini-batch selected from $\mathcal{D}_k$. Next, all devices report the local gradients in (\ref{e3}) to the PS.
Upon receiving all the local gradients, PS performs the update as
\begin{align}\label{e4}
\bm{w}_{m+1}=\bm{w}_{m}-\eta \bm{g}_m,
\end{align}
where $\eta$ denotes the learning rate and
\begin{align}\label{gm}
\bm{g}_m\triangleq \frac{1}{K}\sum_{k=1}^K \bm{g}_m^k.
\end{align}
The FL algorithm iterates (\ref{e3}) and (\ref{e4}) until convergence.

\begin{figure}[!t]
\centering
\includegraphics[width=4.5in]{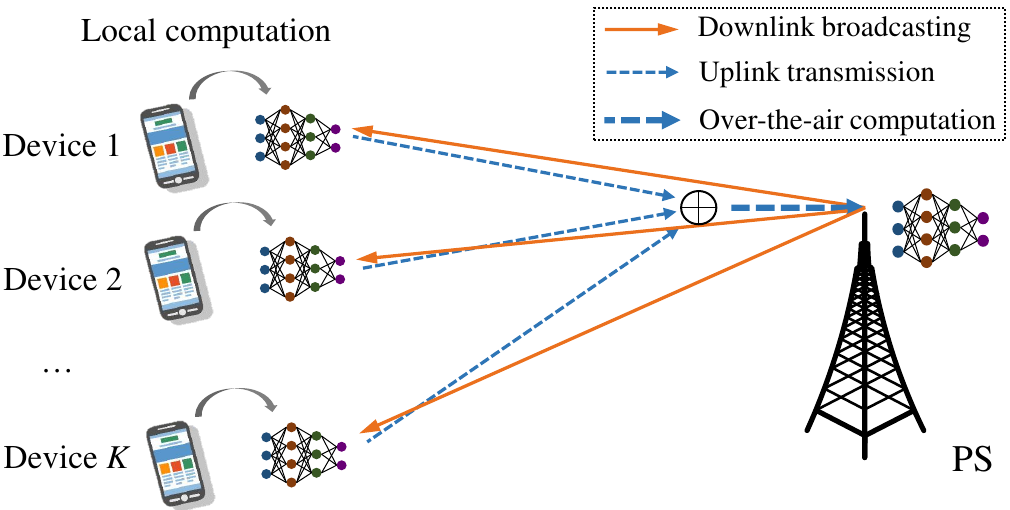}
\caption{An architecture of AirFL with  one PS and $K$ devices.}\label{fig1}
\end{figure}

\subsection{Over-the-Air Computation for FL}
In practice, we adopt the AirComp method for model aggregation in wireless networks, shown in Fig. \ref{fig1}. We express the channel between the $k$th device and the PS as $d_k^{-\frac{\alpha}{2}}h_k$, where $d_k$ denotes the distance between the PS and device $k$, $\alpha$ is the large scale path loss exponent, and $h_k$ represents the small-scale fading of the channel. Assume that the channels are independent Rayleigh fading channels, i.e., $h_k\sim \mathcal{CN}(0,1)$. In general, the small-scale fading of the channel cannot be perfectly estimated at devices. Denote the channel estimate at device $k$ by $\hat{h}_k$ and a relationship between $h_k$ and $\hat{h}_k$ can be modelled as 
$h_k = \rho \hat{h}_k + \sqrt{1-\rho^2}v_k$, 
where $\rho\in(0,1]$ represents the correlation coefficient between $h_k$ and $\hat{h}_k$, and $v_k\sim \mathcal{CN}(0,1)$ is the error independent of $\hat{h}_k$. Note that $\rho$ directly corresponds to the level of channel estimation accuracy and $\rho=1$ implies the availability of perfect CSI. To overcome the negative impact of deep fading, a must truncated channel inversion scheme is headed for the uplink transmission \cite{zhu}. To be concrete, only when $\vert \hat{h}_k \vert^2 $ exceeds a predetermined threshold, $\gamma_{\mathrm{th}}$, the device is activated to transmit its gradient to PS. Accordingly, the received signal at the PS follows
\begin{align}
\bm{y}=\sum_{k\in \mathcal{S}_m} d_k^{-\frac{\alpha}{2}} h_k \beta_k \bm{g}_m^k +\bm{z}_m,
\end{align} 
where $\mathcal{S}_m$ represents the set of activated devices, $\beta_k$ is the pre-processing factor for device $k$, and $\bm{z}_m\sim \mathcal{CN}(\bm{0},\sigma^2\bm{I})$ is the additive Gaussian noise with power $\sigma^2$. Given the imperfect CSI, the pre-processing factor for device $k$ is chosen as $\beta_k =\frac{\zeta \lambda d_k^{\alpha/2} \hat{h}_k^*}{K \vert \hat{h}_k \vert^2}$ \cite{guo}, where $\zeta$ is a scaling factor for ensuring the transmit power constraint and $\lambda$ is a compensation constant for ensuring unbiasedness of the gradient estimation. For simplicity, we consider the uniform transmit power budget $P_{\max}$ at each device and choose the factor $\zeta$ to guarantee
\begin{align}\label{pmax}
\mathbb{E}\left [\left \Vert \beta_k \bm{g}_m^k \right \Vert^2 \right]\leq P_{\max}.
\end{align}
At the receiver, by scaling $\bm{y}$ with $\frac{1}{\zeta}$ and taking the real part, an estimate of the actual gradient in (\ref{gm}) is given by
\begin{align} \label{ghat}
\hat{\bm{g}}_m=\frac{1}{K} \sum_{k=1}^K \xi_k \bm{g}_m^k +\bar{\bm{z}}_m,
\end{align}
where $\bar{\bm{z}}_m\!\triangleq \!\frac{\Re\{\bm{z}_m\}}{\zeta}$ is the equivalent noise, and $\xi_k$ is given by
\begin{align}
\xi_k =\left \{  
\begin{array}{cc}
\lambda\frac{\Re\{h_k^* \hat{h}_k\}}{\vert \hat{h}_k \vert^2} & \vert \hat{h}_k \vert^2\geq \gamma_{\mathrm{th}},\\
0&\vert \hat{h}_k \vert^2<\gamma_{\mathrm{th}}.
\end{array}
\right.
\end{align}
By comparing (\ref{ghat}) and (\ref{gm}),  the distortion in the gradient estimation comes from two aspects, i.e., the coefficient distortion $\xi_k$ caused by the imperfect CSI, and the scaled additive Gaussian noise $\bar{\bm{z}}_m$. Also, we notice that the expectation of $\xi_k$ determines whether the gradient estimation is unbiased, and the variance of $\xi_k$ and $\bar{\bm{z}}_m$ measure the gradient estimation distortion, which brings notable deterioration in convergence performance \cite{ajoint}.

\section{Performance Analysis and Optimization}
In this section, we theoretically capture the  impact of the imperfect CSI on the performance of AirFL. Based on the analytical results, we further optimize the truncation threshold $\gamma_{\mathrm{th}}$ with respect to $\rho$ and system SNR.

\vspace{-0.2cm}
\subsection{Performance Analysis for AirFL}

Firstly, we need to determine the value of the compensation constant $\lambda$ for achieving an unbiased gradient estimation.

\begin{lemma}
In order to ensure the unbiasedness of gradient transmission, the compensation constant truncation and imperfect CSI is chosen by $\lambda=\frac{\mathrm{e}^{\gamma_{\mathrm{th}}}}{\rho}$.
\end{lemma}
\begin{proof}
Please refer to Appendix A.
\hfill $\square$
\end{proof}

Then, to facilitate the performance analysis for AirFL, the following \emph{Lemma~2} derives the variance of AirComp parameters $\xi_k$, which directly reflects the mean squared error (MSE) of AirComp \cite{zhang}.

\begin{lemma}
The variance of $\xi_k$ with unit mean is given by
\begin{align}\label{e5}
\mathbb{E}\left[(\xi_k -1)^2  \right ]=\mathrm{e}^{\gamma_{\mathrm{th}}}-\frac{1-\rho^2}{2\rho^2}\mathrm{Ei}(-\gamma_{\mathrm{th}})\mathrm{e}^{2\gamma_{\mathrm{th}}}-1,
\end{align}
where $\mathrm{Ei}(\cdot)$ denotes the exponential integral function.
\end{lemma}
\begin{proof}
Please refer to Appendix B.
\hfill $\square$
\end{proof}

\begin{remark}
Note that $\lim_{\gamma_{\mathrm{th}}\to 0}\mathrm{Ei}(\gamma_{\mathrm{th}})=-\infty$. Hence, from (\ref{e5}), regardless of the level of channel estimation accuracy, the truncation is necessary to avoid unbounded variance. 
\end{remark}

\begin{remark}
It is worth noting that for other power control schemes of AirComp, e.g., that in \cite{zhu2}, the method of theoretical analysis still applies through treating them as special truncated channel inversion schemes. Without loss of generality, we therefore consider the most commonly adopted truncated channel inversion scheme as a general analysis. 
\end{remark}

Next, to evaluate the accuracy loss caused by the imperfect model aggregation, we choose a popular performance metric as the expected weight divergence with respect to $\hat{\bm{g}}_m$ and $\bm{g}_m$ \cite{gomore}, defined by $\Delta^2=\mathbb{E}\left[\left \Vert\hat{\bm{g}}_m- \bm{g}_{m}\right \Vert^2 \right]$. It is worth noting that $\Delta^2$ corresponds to the MSE of the gradient estimation at the PS and directly reflects the accuracy of gradient estimation via the AirComp, which determines the convergence performance.
To pave the way for performance analysis, we make the following widely used assumption \cite{zhu2}.

\emph{Assumption}: The stochastic gradients on random batches are uniformly bounded, i.e., $\mathbb{E}\left[\left \Vert \bm{g}_m^k \right \Vert^2\right]\leq G^2$. And the obtained global gradient, $\bm{g}_m$, is unbiased and variance bounded, i.e.,
\begin{align}
&\mathbb{E}\left [ \bm{g}_m \right] =\nabla F (\bm{w}_m), \enspace \mathbb{E}\left [ \left \Vert \bm{g}_m -\nabla F (\bm{w}_m)\right \Vert \right] \leq \delta^2.
\end{align}
Then, the weight divergence can be accurately characterized under this general assumption.

\begin{theorem}
The weight divergence, $\Delta^2$, is bounded by
\begin{align}\label{e11}
\Delta^2\!\leq\!\frac{G^2}{K^2}\!\!\left(\!\!\mathrm{e}^{\gamma_{\mathrm{th}}}\!-\!\frac{1\!-\!\rho^2}{2\rho^2}\mathrm{Ei}(\!-\gamma_{\mathrm{th}}\!)\mathrm{e}^{2\gamma_{\mathrm{th}}}\!-\!1\!+\!\frac{  \sigma^2\!\max_k\!\left \{\!d_k^\alpha\!\right\}\mathrm{e}^{2\gamma_{\mathrm{th}}}}{ 2 P_{\max} \rho^2  \gamma_{\mathrm{th}}}\!\right)\!.
\end{align}
\end{theorem}
\begin{proof}
Please refer to Appendix C.\hfill $\square$
\end{proof}
\begin{remark}
According to (\ref{e11}), the imperfect channel estimation deteriorates the learning performance with the order of $\frac{1}{\rho^2}$. It validates our statement that the accurate channel estimation is a key to the AirFL system.
\end{remark}
\begin{remark}
Especially for high SNR regime, i.e., $\frac{P_{\max}}{\sigma^2}\to \infty$, the weight divergence $\Delta^2$ is dominated by the impact of imperfect CSI rather than noise. It implies that, for a given level of channel estimation accuracy, the accuracy loss caused by imperfect CSI can no longer be compensated by increasing the transmit power while only weakens the impact of the noise. This is the key observation that is different from the impact of CSI error in pure communication systems for data recovery.
\end{remark}

\begin{remark}
By direct inspection of (\ref{e11}), as the increase of the number of devices, $K$, the weight divergence decreases as $\mathcal{O}\left( 1/K^2\right)$ and eventually tends towards zero, i.e., the impact of imperfect CSI is completely eliminated. This phenomenon can be qualitatively explained by the law of Large Numbers, that is, the randomness of aggregation distortion is eliminated when the participating devices tend to be infinite many. 
\end{remark}

Then, starting from (\ref{e11}), the convergence performance of FL is characterized in the following theorem.

\begin{theorem}
    Suppose the loss function $F$ is $L$-Lipschitz with respect to $\bm{w}$ and the learning rate satisfies $\eta<\frac{2}{L}$. The convergence of FL algorithm is bounded by
    \begin{align}
        \frac{1}{M}\sum_{m=0}^{M-1} \mathbb{E}\left[\left\Vert \nabla F(\bm{w}_m) \right \Vert^2 \right] \leq \frac{1}{M}\left(\frac{F(\bm{w}_0)-\mathbb{E}\left[ F(\bm{w}_{M})\right]}{\eta-\frac{L\eta^2}{2}} +\frac{ML\eta (\Delta^2+\delta^2)}{2-L\eta} \right). 
    \end{align}
\end{theorem}
\begin{proof}
    Please refer to Appendix D.\hfill $\square$
\end{proof}

This theorem implies that the convergence is guaranteed with sufficiently large $M$ and the gap to the optimality converges to $\frac{L\eta (\Delta^2+\delta^2)}{2-L\eta}$, which linearly increasing with respect to $\Delta^2$.

\subsection{Optimization of the Truncation Threshold}

According to the result in \emph{Theorem 1}, we find that he impacts of learning algorithms and the wireless transmission are decoupled. Hence, the truncation threshold optimization can be isolated from the specific learning algorithms and parameters, thus being defined as follows:
\begin{align}\label{opt}
\max_{\gamma_{\mathrm{th}}>0}\quad h(\gamma_{\mathrm{th}})\triangleq  \mathrm{e}^{\gamma_{\mathrm{th}}}-k_1 \mathrm{Ei}(-\gamma_{\mathrm{th}})\mathrm{e}^{2\gamma_{\mathrm{th}}}+k_2\frac{\mathrm{e}^{2\gamma_{\mathrm{th}}}}{\gamma_{\mathrm{th}}},
\end{align}
where $k_1\!\triangleq\!\frac{1-\rho^2}{2\rho^2}$, and $k_2\!\triangleq\!\frac{\sigma^2 \max_k \left \{d_k^\alpha \right\}}{2P_{\max} \rho^2}$ are positive~constants.

\begin{theorem}
The objective function in (\ref{opt}) is convex.
\end{theorem}
\begin{proof}
Please refer to Appendix E. \hfill $\square$
\end{proof}

Based on convexity of $h(\cdot)$, the optimal value of $\gamma_{\mathrm{th}}$ can be easily obtained from a bisection method with low complexity. Specifically, the  derivative of $h(x)$ is
\begin{align}
h^\prime(x)=\mathrm{e}^x-k_1 \frac{\mathrm{e}^x}{x}-2k_1\mathrm{Ei}(-x)\mathrm{e}^{2x}+k_2 \mathrm{e}^{2x}\frac{2x-1}{x^2}.
\end{align}
Since $\lim_{x\to0}h^\prime(x)<0$ and $\lim_{x\to\infty}h^\prime(x)>0$, the unique zero point of $h^\prime(x)$, i.e., the optimal solution of $\gamma_{\mathrm{th}}$, can be found through the bisection search.

\section{Simulation Results}
In this section, we provide simulation results to verify the performance analysis and truncation threshold optimization. We train a multi-layer perceptron (MLP) on the popular MNIST dataset via the AirFL algorithm. The distance $d_k$ is uniformly distributed over $(0,500)\,$m. Unless otherwise specified, the other parameters are set as: $K=10$, $\alpha=2.2$, $P_{\max}=0.1$ W, $\sigma^2=-40$ dBm, and  $\eta=0.005$.

Fig. \ref{fig2} depicts the numerical variance of $\xi_k$ obtained from Monte-Carlo simulations, compared with the theoretical result in (\ref{e5}). It shows that the numerical results matches well with the theoretical results, which verifies our analysis. Moreover, the channel estimation accuracy level, $\rho$, imposes more significant impacts on the variance than the truncation threshold.

\begin{figure}[!t]
\centering
\includegraphics[width=6in]{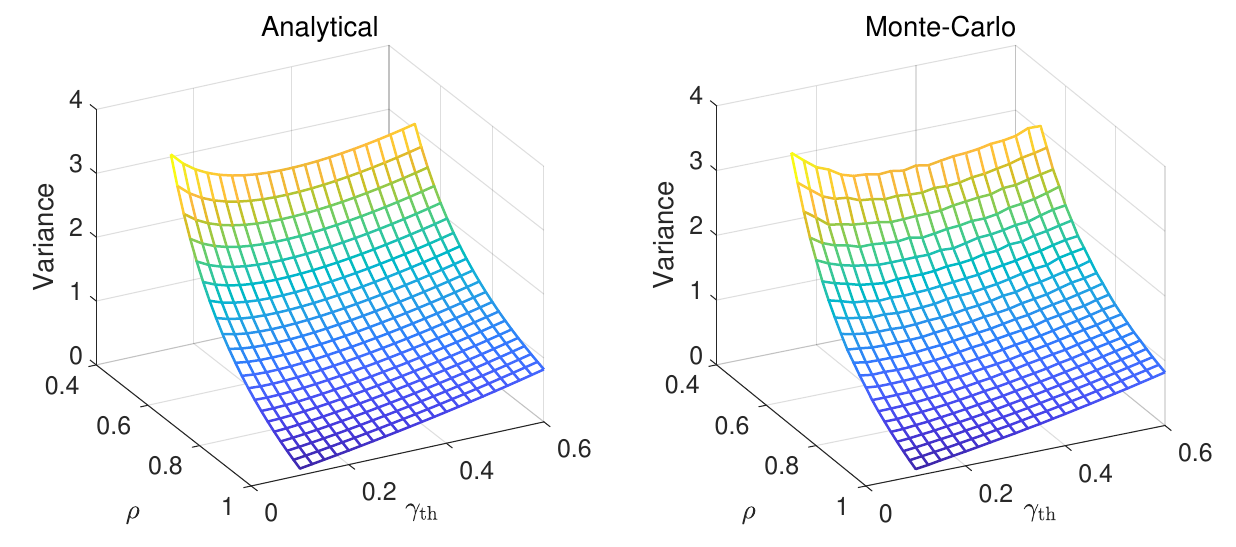}
\caption{ Variance versus $\rho$ and $\gamma_{\mathrm{th}}$.}\label{fig2}
\vspace{-0.45cm}
\end{figure}
\begin{figure}[!t]
\centering
\includegraphics[width=6in]{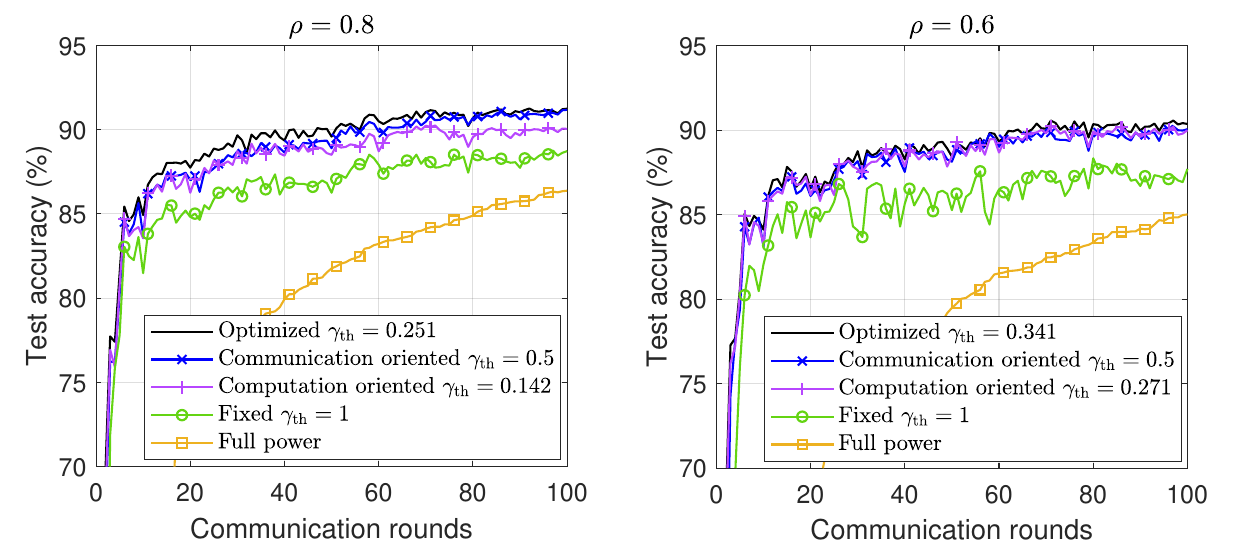}
\caption{ Test accuracy versus different truncation thresholds.}\label{fig3}
\vspace{-0.45cm}
\end{figure}
In Fig. \ref{fig3}, we evaluate the impact of truncation threshold optimization on learning performance and compare it with other schemes.  The four benchmark schemes are described as: ``Communication oriented" and ``Computation oriented" schemes represent $\gamma_{\mathrm{th}}$ is optimized to minimize the noise related and computation related term in (\ref{e11}), respectively; ``Fixed $\gamma_{\mathrm{th}}$" represents  the truncation threshold is set as a constant \cite{zhu}; The ``Full power" scheme represents the transmitter does not perform power control and only compensates channel phase offset \cite{fullp}. For all the tested setups, the proposed optimization method outperforms all the benchmarks due to the joint consideration of communication and computation. It is observed that the test accuracy first improves and then deteriorates with $\gamma_{\mathrm{th}}$. This is because with the increase of $\gamma_{\mathrm{th}}$, the performance is first limited by noise and then limited by less participating devices. Also, for larger $\rho$, a lower truncation threshold is preferred. Moreover, compared with the full power scheme, the proposed power control strategy successfully alleviates the impact of data heterogeneity, leading to a much prominent performance gain.

\section{Conclusion}
In this paper, we theoretically analyzed the performance of AirFL with imperfect CSI and optimized a channel truncation strategy. The analytical results revealed the importance of accurate channel estimation for AirFL. Our results can also be extended to performance analysis and optimization for other power control schemes of AirFL.

\appendices
\section{Proof of Lemma 1}
To determine $\lambda$, we start with the expectation of $\xi_k$, expressed as 

\begin{align}\label{a0}
\mathbb{E}\left [\xi_k\right ]\!=\!\lambda \mathbb{E}\left [\! \left.\frac{\Re\{h_k^* \hat{h}_k\}}{\vert \hat{h}_k \vert^2} \right \vert \vert \hat{h}_k \vert^2\geq \gamma_{\mathrm{th}} \!\right ]\!\Pr\left \{ \!\vert \hat{h}_k \vert^2\geq \gamma_{\mathrm{th}}\!\right\},
\end{align}
which should be equal to $1$ to guarantee an unbiased gradient estimation in (\ref{ghat}).
Considering that  $h_k$ and its estimate $\hat{h}_k$ are correlated, we first introduce a new random variable to tackle with this difficulty, which follows
\begin{align}\label{a1}
x\triangleq \frac{1}{\sqrt{1-\rho^2}}\left ( \frac{\Re\{h_k^* \hat{h}_k\}}{\vert \hat{h}_k \vert^2} -\rho \right)=  \frac{\Re\{v_k^* \hat{h}_k\}}{\vert \hat{h}_k \vert^2},
\end{align}
where $v_k$ and $\hat{h}_k$ are uncorrelated Gaussian variables with zero mean and unit variance. Then, we have
\begin{align}
\mathbb{E}\left [ x\left \vert \vert \hat{h}_k \vert^2\geq \gamma_{\mathrm{th}} \right.\!\right ]\!=\!\mathbb{E}\!\left [ \!\left.\frac{v_k^* \hat{h}_k+v_k \hat{h}_k^*}{2\vert \hat{h}_k \vert^2} \right \vert \vert \hat{h}_h \vert^2\!\geq\! \gamma_{\mathrm{th}} \!\right ]\!=\!0.
\end{align}
Then, by comparing (\ref{a0}) and (\ref{a1}), through some linear transformations, we arrive at
$\mathbb{E}\left [\xi_k\right ]= \lambda \mathrm{e}^{-\gamma_{\mathrm{th}}}\rho=1$,
which implies that $\lambda=\mathrm{e}^{\gamma_{\mathrm{th}}}/\rho$ and the proof completes. 

\section{Proof of Lemma 2}
We derive the variance of $\xi_k$ by using the form of conditional expectation as 
\begin{align} \label{a3}
&\mathbb{E}\left [(\xi_k -1)\right ]^2 = \mathbb{E}\left [ \left.\left(\frac{\Re\{h_k^* \hat{h}_k\}\mathrm{e}^{\gamma_{\mathrm{th}}}}{\vert \hat{h}_k \vert^2 \rho}-1\right)^2 \right \vert \vert \hat{h}_k \vert^2\geq \gamma_{\mathrm{th}} \right ] \Pr\left \{ \vert \hat{h}_k \vert^2\geq \gamma_{\mathrm{th}}\right\}+\Pr\left \{ \vert \hat{h}_k \vert^2< \gamma_{\mathrm{th}}\right\}\nonumber \\
&\quad= \frac{\mathrm{e}^{\gamma_{\mathrm{th}}}\!(1\!-\!\rho^2)}{\rho^2}\mathbb{E}\!\left [\! \left.\left(x\!-\!c\right)^2 \right \vert y\!\leq\! -\gamma_{\mathrm{th}} \!\right ]+1-\!\mathrm{e}^{-\gamma_{\mathrm{th}}},
\end{align}
where $y\!\triangleq\! -\vert \hat{h}_k \vert^2$ and $c\!\triangleq \!\frac{\rho\left(1-\mathrm{e}^{\gamma_{\mathrm{th}}}\right)}{\mathrm{e}^{\gamma_{\mathrm{th}}}\sqrt{1-\rho^2}}$. To calculate the conditional expectation, we first need to find the joint distribution of $x$ and $y$. The joint cumulative distribution function (CDF) of  $x$ and $y$ equals
\begin{align}\label{a2}
F_{xy}(t,\gamma)&=\Pr \left \{ \frac{\Re\{v_k^* \hat{h}_k\}}{\vert \hat{h}_k \vert^2} <t,\,-\vert \hat{h}_k \vert^2 <\gamma \right \} \nonumber \\
&=\Pr \left \{ v_k^* \hat{h}_k+ v_k \hat{h}_k^* -2t\vert \hat{h}_k \vert^2<0, \,-\vert \hat{h}_k \vert^2 <\gamma \right \} \nonumber \\
&=\Pr \left \{ \bm{z}^H \bm{A}_1 \bm{z} <0,\, \bm{z}^H \bm{A}_2 \bm{z} <\gamma \right \},
\end{align}
where $\bm{z}\triangleq [\hat{h}_k, v_k]^H$, and 
\begin{align}
\bm{A}_1\triangleq \left [
\begin{array}{cc}
-2t& 1\\1 & 0
\end{array}\right],\enspace 
\bm{A}_2\triangleq \left [
\begin{array}{cc}
-1& 0\\0 & 0
\end{array}\right].
\end{align}
According to \cite[Eq. (3.2c.5)]{quad}, the joint moment generating function (MGF) of the two quadratic forms, $z_1\triangleq \bm{z}^H \bm{A}_1 \bm{z}$ and $z_2\triangleq \bm{z}^H \bm{A}_2 \bm{z}$, follows
\begin{align}
\mathcal{M}_{z_1,z_2}(s_1,s_2)=\mathrm{det} \left ( \bm{I}-s_1 \bm{A}_1 -s_2\bm{A}_2 \right)^{-1}.
\end{align}
Applying the inverse Laplace transformation, we express the probability in (\ref{a2}) as
\begin{align}\label{a23}
\Pr &\left \{ \bm{z}^H \bm{A}_1 \bm{z} <0,\, \bm{z}^H \bm{A}_2 \bm{z} <\gamma \right \}\nonumber \\
&=\frac{1}{(2\pi i)^2} \int_{\epsilon_1-i\infty }^{\epsilon_1+i\infty} \int_{\epsilon_2-i\infty }^{\epsilon_2+i\infty} \frac{\mathrm{e}^{\gamma s_2 }}{s_1 s_2} \mathcal{M}_{z_1,z_2}(s_1,s_2) \mathrm{d} s_2 \mathrm{d} s_1 \nonumber \\
&\overset{\text{(a)}}{=}  \frac{1}{2\pi i}\int_{\epsilon_1-i\infty }^{\epsilon_1+i\infty} \frac{1}{s_1(1+2ts_1-s_1^2)}\mathrm{d} s_1 - \frac{1}{2\pi i}\int_{\epsilon_1-i\infty }^{\epsilon_1+i\infty} \frac{1}{s_1(1+2ts_1-s_1^2)}\mathrm{e}^{-(1+2ts_1-s_1^2)\gamma}\mathrm{d} s_1
\nonumber \\
&\overset{\text{(b)}}{=}\frac{t+\sqrt{1+t^2}}{2\sqrt{1+t^2}}+   \frac{1}{2\pi i}\int_{\epsilon-i\infty }^{\epsilon+i\infty} \frac{\left(t+\sqrt{t^2+s+1}\right)\mathrm{e}^{\gamma s}}{2s(s+1)\sqrt{t^2+s+1}}\mathrm{d} s \nonumber
\nonumber \\
&\overset{\text{(c)}}{=}\frac{t}{2\sqrt{1+t^2}}\left (1-\mathrm{Erf}\left(\sqrt{-\gamma(1+t^2)} \right) U(-\gamma)\right)+\frac{\mathrm{e}^{\gamma}}{2}\mathrm{Erfc}\left(-\sqrt{-\gamma}t\right)U(-\gamma),
\end{align}
where (a) comes from Eq. (5.2.4) in supplements of \cite{laplace}, (b) exploits the Cauchy's residue theorem and $s\triangleq s_1^2 -2ts_1-1$, (c) is due to Eq. (5.3.7) in supplements of \cite{laplace}, $\mathrm{Erf}(\cdot)$, $\mathrm{Erfc}(\cdot)$, and $U(\cdot)$ represent the error function, the complementary error function and the Heaviside function, respectively. From (\ref{a2}) and (\ref{a23}) and by taking the derivative of joint CDF, the joint probability density function (PDF) of $x$ and $y$ equals
\begin{align}\label{pdf}
f_{xy} (t,\gamma)=\sqrt{-\frac{\gamma}{\pi}}\mathrm{e}^{\gamma(1+t^2)}U(-\gamma).
\end{align}
From (\ref{pdf}), we calculate the conditional expectation in (\ref{a3}) as 
\begin{align}\label{a4}
\mathbb{E}&\left [ \left.\left(x-c\right)^2 \right \vert y\leq -\gamma_{\mathrm{th}} \right ]
=\int_{-\infty}^{\infty} \frac{\int_{-\infty}^{-\gamma_{\mathrm{th}}}f_{xy} (t,\gamma) \mathrm{d}\gamma }{\Pr\left \{ y\leq -\gamma_{\mathrm{th}}\right \}} \mathrm{d}t \nonumber \\
&=\mathrm{e}^{\gamma_{\mathrm{th}}} \int_{-\infty}^{-\gamma_{\mathrm{th}}}\int_{-\infty}^{\infty}(t-c)^2\sqrt{-\frac{\gamma}{\pi}}\mathrm{e}^{\gamma(1+t^2)}\mathrm{d}t \mathrm{d}\gamma \nonumber \\
&\overset{\text{(a)}}{=}\mathrm{e}^{\gamma_{\mathrm{th}}} \int_{-\infty}^{-\gamma_{\mathrm{th}}} \frac{2c^2\gamma-1}{2\gamma}\mathrm{e}^{\gamma}\mathrm{d}\gamma \nonumber \\
&\overset{\text{(b)}}{=}  c^2 -\frac{1}{2}\mathrm{Ei}(-\gamma_{\mathrm{th}})\mathrm{e}^{\gamma_{\mathrm{th}}},
\end{align}
where (a) exploits \cite[Eq. (3.462.8)]{table} and \cite[Eq. (3.321.1)]{table} and the fact that $\int_{-\infty}^{\infty}t\mathrm{e}^{\gamma t^2}\mathrm{d}t=0$. The equality in (b) comes from the definition of the exponential integral function, $\mathrm{Ei}(\cdot)$. Applying (\ref{a4}) into (\ref{a3}), we complete the proof.

\section{Proof of Theorem 1}
The weight divergence is reformulated as 
\begin{align}\label{a5}
\mathbb{E}&\left[ \left \Vert \hat{\bm{g}}_m -\bm{g}_m \right \Vert^2 \right]=\mathbb{E}\left[ \left \Vert \frac{1}{K}\sum_{k=1}^K (\xi_k-1)\bm{g}_m^k +\bar{\bm{z}}_m \right \Vert^2 \right]\nonumber \\
&\overset{\text{(a)}}{=}\frac{1}{K^2}\sum_{k=1}^K \mathbb{E}\left[(\xi_k-1)^2 \right ]\mathbb{E}\left[ \left \Vert\bm{g}_m^k \right \Vert^2 \right] +\mathbb{E}\left[ \left \Vert \bar{\bm{z}}_m \right \Vert^2\right ],
\end{align}
where (a) is due to the zero mean and independence between $\xi_k -1$, $\forall k$.  As for the noise term, recall that $\bar{\bm{z}}_m=\frac{\Re\{\bm{z}_m\}}{\zeta}$ and we have $\mathbb{E}\left[ \left \Vert \bar{\bm{z}}_m \right \Vert^2\right ]=\frac{\sigma^2}{2\zeta^2}$.
According to the transmit power constraint in (\ref{pmax}), the scaling factor $\zeta$ must satisfy
\begin{align}\label{a28}
\max_{k\in \mathcal{S}_m} \left\{\frac{\zeta^2 \lambda^2 d_k^{\alpha} }{K^2 \vert \hat{h}_k \vert^2}\mathbb{E}\left[ \left \Vert\bm{g}_m^k \right \Vert^2 \right] \right\} \leq P_{\max}. 
\end{align}
Note that for all $k\!\in\!\mathcal{S}_m$, we have $\vert \hat{h}_k \vert^2\!\geq\! \gamma_{\mathrm{th}}$. Combining (\ref{a28}) with the value of $\lambda$ and the bound assumption of $\mathbb{E}\left[\! \left \Vert\bm{g}_m^k \right \Vert^2 \!\right]\!\leq\!G^2$, $\zeta$ is set as $\zeta\!=\!\!\frac{K \rho \sqrt{P_{\max}\gamma_{\mathrm{th}}}} {G\!\max_k\!\left \{d_k^{\alpha/2}\right\}\mathrm{e}^{\gamma_{\mathrm{th}}}}$.
Then, combining all the derived results, we obtain (\ref{e11}) and complete the proof.

\section{Proof of Theorem 2}
Under the general assumption, we have
\begin{align} \label{eq3}
    \mathbb{E}&\left[F(\bm{w}_{m+1})-F(\bm{w}_m)\right]\nonumber \\
    &\overset{\text{(a)}}{\leq} \mathbb{E}\left[-\eta (\nabla F(\bm{w}_m))^T \hat{\bm{g}}_m+\frac{L\eta^2 }{2} \Vert \hat{\bm{g}}_m \Vert^2 \right ] \nonumber \\
    &\overset{\text{(b)}}{=}-\eta \mathbb{E}\left[\left\Vert \nabla F(\bm{w}_m) \right \Vert^2 \right] +\frac{L\eta^2}{2} \mathbb{E}\left[\left\Vert \hat{\bm{g}}_m -\bm{g}_m+ \bm{g}_m-\nabla F(\bm{w}_m)+\nabla F(\bm{w}_m) \right \Vert^2\right] \nonumber \\
    &\overset{\text{(c)}}{=}-\left(\eta-\frac{L\eta^2}{2}\right)\mathbb{E}\left[\left\Vert \nabla F(\bm{w}_m) \right \Vert^2 \right] + \frac{L\eta^2}{2} \mathbb{E}\left[ \left\Vert \hat{\bm{g}}_m -\bm{g}_m\right \Vert^2\right]+\frac{L\eta^2}{2} \mathbb{E}\left[\left\Vert  \bm{g}_m-\nabla F(\bm{w}_m) \right \Vert^2\right] \nonumber \\
    &\overset{\text{(d)}}{\leq} -\left(\eta-\frac{L\eta^2}{2}\right)\mathbb{E}\left[\left\Vert \nabla F(\bm{w}_m) \right \Vert^2 \right]+ \frac{L\eta^2 (\Delta^2+\delta^2)}{2},
\end{align}
where (a) is due to the fact that $F(\cdot)$ is $L$-Lipschitz and the definition of $\bm{w}_{m+1}$, (b) comes from \emph{Lemma 1}, (c) uses the assumption, and (d) exploits \emph{Theorem 1}.
By summing (\ref{eq3}) from $m=0$ to $m=M-1$, we have 
\begin{align}
    \frac{1}{M}\sum_{m=0}^{M-1} \mathbb{E}\left[\left\Vert \nabla F(\bm{w}_m) \right \Vert^2 \right] \leq \frac{1}{M}\left(\frac{F(\bm{w}_0)-\mathbb{E}\left[ F(\bm{w}_{M})\right]}{\eta-\frac{L\eta^2}{2}} +\frac{ML\eta (\Delta^2+\delta^2)}{2-L\eta} \right),
\end{align}
which holds for any small $\eta<\frac{2}{L}$. 

\section{Proof of Theorem 3}
We check the second derivative of $h(x)$ as 
\begin{align}\label{e30}
h^{\prime\prime}(x)=& \mathrm{e}^{x}+ \frac{k_1 \mathrm{e}^{x}}{x^2}\left(-4x^2 \mathrm{e}^{x}\mathrm{Ei}(-x)-3x+1\right)+\frac{2k_2\mathrm{e}^{x}(2x^2-2x+1)}{x^3}.
\end{align}
It is obvious that the first and the third terms in (\ref{e30}) are positive for positive $x$. To prove the nonnegativity of the second term, we use the inequalities \cite[Eq. (5.1.20)]{ei}
\begin{align}
-\mathrm{Ei}(-x)>\frac{1}{2}\mathrm{e}^{-x}\mathrm{ln}\left(1+\frac{2}{x}\right)> \mathrm{e}^{-x}\frac{1}{x+1},
\end{align}
which yields to
$-4x^2 \mathrm{e}^{x}\mathrm{Ei}(-x)-3x+1>\frac{(x-1)^2}{x+1}\geq 0$.
Then, we conclude that $h^{\prime\prime}(x)>0$ for $x>0$ and hence $h(x)$ is convex.

\end{document}